\documentclass[10pt,conference]{IEEEtran}
\IEEEoverridecommandlockouts
\usepackage{authblk}
\usepackage{cite}
\usepackage{amsmath,amssymb,amsfonts,dsfont}
\usepackage{algorithm}
\usepackage[noend]{algorithmic}
\usepackage{graphicx}
\usepackage{textcomp}
\usepackage{xcolor}
\usepackage{bm}
\usepackage{booktabs}
\usepackage{multirow}
\usepackage{lipsum}
\usepackage{url}
\usepackage{amsthm}
\usepackage{subcaption}

\ifCLASSOPTIONcompsoc
\usepackage[caption=false, font=normalsize, labelfont=sf, textfont=sf]{subfig}
\else
\usepackage[caption=false, font=footnotesize]{subfig}

\newtheorem{theorem}{Theorem}

\newtheorem{lemma}{Lemma}

\newcommand{\mbbR}{\mathbb{R}} 

\newcommand{\bdd}{\boldsymbol{d}} 
\newcommand{\bdr}{\boldsymbol{r}} 
\newcommand{\bdw}{\boldsymbol{w}}

\newcommand{\bdz}{\boldsymbol{z}}

\newcommand{\bdalpha}{\boldsymbol{\alpha}} 
\newcommand{\bdbeta}{\boldsymbol{\beta}}


\allowdisplaybreaks[4] 

\begin{document}

\title
{
Efficient and Provably Convergent Computation of Information Bottleneck: A Semi-Relaxed Approach 
\thanks{The first three authors contributed equally to this work and $\dag$ marked the corresponding author. This work was partially supported by National Key Research and Development Program of China (2018YFA0701603) and National Natural Science Foundation of China (12271289 and 62231022).}
}

\author[1]{Lingyi Chen}
\author[1]{Shitong Wu}
\author[1]{Jiachuan Ye}
\author[2]{Huihui Wu}
\author[3]{Wenyi Zhang}
\author[1$\dag$]{Hao Wu}
\affil[1]{Department of Mathematical Sciences, Tsinghua University, Beijing 100084, China}
\affil[2]{Yangtze Delta Region Institute (Huzhou), 
\authorcr
University of Electronic Science and Technology of China, Huzhou, Zhejiang 313000, China}
\affil[3]{Department of Electronic Engineering and Information Science, 
\authorcr University of Science and Technology of China, Hefei, Anhui 230027, China 
\authorcr Email: hwu@tsinghua.edu.cn}

\maketitle

\begin{abstract}
Information Bottleneck (IB) is a technique to extract information about one target random variable through another relevant random variable. This technique has garnered significant interest due to its broad applications in information theory and deep learning. Hence, there is a strong motivation to develop efficient numerical methods with high precision and theoretical convergence guarantees. In this paper, we propose a semi-relaxed IB model, where the Markov chain and transition probability condition are relaxed from the relevance-compression function. Based on the proposed model, we develop an algorithm, which recovers the relaxed constraints and involves only closed-form iterations. Specifically, the algorithm is obtained by analyzing the Lagrangian of the relaxed model with alternating minimization in each direction. The convergence property of the proposed algorithm is theoretically guaranteed through descent estimation and Pinsker's inequality. Numerical experiments across classical and discrete distributions corroborate the analysis. Moreover, our proposed algorithm demonstrates notable advantages in terms of computational efficiency, evidenced by significantly reduced run times compared to existing methods with comparable accuracy. 
\end{abstract}

\IEEEpeerreviewmaketitle

\section{Introduction}
The Information Bottleneck (IB) theory, initially introduced by Tishby \cite{tishby_1999_ib}, offers a method to extract relevant information, represented by a variable $Y$, from an observable variable $X$, without using explicit distortion measures \cite{cover_2003_elements}. 
This extracted information is represented by a bottleneck variable $T$, forming a Markov chain $Y\leftrightarrow X\leftrightarrow T$. 
The primary goal of the IB problem is to identify a bottleneck variable $T$ that minimizes $I(T;X)$ while keeping $I(T;Y)$ above a prescribed threshold. 
This concept is instrumental in characterizing the balance between lossy compression rates and distortion thresholds \cite{tishby_2003_tradeoff}. 
Specifically, given a source random variable $X$ with probability distribution $P_X$, and a relevant variable $Y$ with conditional distribution $P_{Y|X}$, the IB objective, also called the relevance-compression function (denoted as RI), is defined by
\begin{equation}
R(I):=\min_{P_{T|X}}~I(T;X),\quad\text{s.t.}\quad I(T;Y)\geq I.
\label{def_RI}
\end{equation}

The application of the IB theory extends widely, encompassing areas such as information theory \cite{hassanpour_2020_forwardaware,chen_2016_alternating,zaidi_2020_models} and machine learning \cite{shamir_2010_generalization,tishby_2017_dnn,alemi2016variation_ib}.
Hence, there is a strong motivation to develop efficient numerical methods with high precision and theoretical convergence guarantees. 
Despite its great importance, devising numerical algorithms for the IB problem that possesses accuracy, efficiency, and convergence guarantees is not straightforward. 
The Blahut-Arimoto (BA) algorithm \cite{tishby_1999_ib} was proposed earlier to compute the RI function given a fixed multiplier, which geometrically corresponds to the tangent slope \cite{zaidi_2020_models} of the RI curve. 
However, it needs to explore the RI curve by sweeping through a range of the multipliers to compute $R(I)$ given a target $I$, which incurs a heavy computational cost and causes instability around linear segments \cite{kolchinsky2018ib_caveats}. 
Besides, when phase transitions occur \cite{wu2019phase}, it becomes challenging to characterize the information plane, since the number of iterations for convergence will increase dramatically.

Recently, some attempts \cite{huang_2021_admm} based on the alternating direction method of multipliers (ADMM) \cite{boyd2011distributed} have been proposed to solve the augmented Lagrange multiplier (ALM) problem, whose convergence is guaranteed by introducing proximal terms. 
However, due to the complexities inherent in subproblem procedures, these methods demonstrate inefficiency and numerical instability when implemented by gradient descent. 
Hence, they are primarily suitable for addressing discrete cases with limited scale.
To enhance efficiency in numerical computation, a method based on the optimal transport (OT) approach has been proposed \cite{Chen23IBOT}, which capitalizes on the OT structure and employs the Sinkhorn algorithm \cite{wu_2023_communication,curuti_2013_sinkhorn}, incorporating closed-form iterations for the computation of the IB problem. 
Although the OT-based methods are efficient and feasible to cases with large scale, its convergence is not theoretically guaranteed. 

To address the aforementioned difficulty, in this work, we propose an accurate, efficient and convergence-guaranteed algorithm for computing the RI function. 
The proposed algorithm is based on a new IB problem formulation by relaxing the Markov chain and the inherent one-side marginal distribution representations, for which we term it a semi-relaxed IB model. 
Based on this semi-relaxed model, we develop an algorithm which recovers the relaxed constraint by analyzing its Lagrangian in an alternative manner. 
The descent value in each iteration is precisely calculated and estimated through the Pinsker's inequality \cite{cover_2003_elements} with our proposed algorithm. 
Moreover, each iteration step only involves a closed-form iteration and the convergence of our algorithm is theoretically guaranteed. 
Numerical experiments show that the proposed method is more efficient and accurate than other existing methods. 
Besides, our method is still efficient and robust in the examples with phase transition phenomena and in some data-driven cases in supervised classification scenarios \cite{amjad2019classification,goldfeld_2020_ml}, where the predicted random variable is a deterministic function of the input observation variable.

\section{Problem Formulation}

Consider a discrete memoryless source $X\in\mathcal{X}$ with a representation $T\in\mathcal{T}$, and a relevant variable $Y\in\mathcal{Y}$, where $\mathcal{X}=\{x_1,\cdots,x_M\}$, $\mathcal{T}=\{t_1,\cdots,t_N\}$, $\mathcal{Y}=\{y_1,\cdots,y_K\}$.
The variables $Y\leftrightarrow X\leftrightarrow T$ form a Markov chain, i.e., 
\begin{equation*}
P_{Y|T}(y_{k}|t_{j})=\frac{\sum_{i=1}^{M}P_{T\mid X}(t_{j}\mid x_{i})P_{Y\mid X}(y_{k}\mid x_{i})P_{X}(x_{i})}{P_{T}(t_{j})},
\end{equation*}
where $P_{T}(t_{j}):=\sum_{i=1}^{M} P_{T\mid X}(t_{j}\mid x_{i})P_{X}(x_{i})$ is the bottleneck distribution $P_{T}$ defined via conditional probability.
To simplify the notation, we denote $p_i=P_X(x_i)$, $s_{ki}=P_{Y|X}(y_k|x_i)$, $w_{ji}=P_{T|X}(t_j|x_i)$, $q_k=P_Y(y_k)$ and $\hat{I}=I+\sum_{k=1}^K q_k\log q_k$.
Then, the RI function defined in \eqref{def_RI} can be written as 
\begin{subequations}\label{org_RI}
\begin{align}
\min_{\bdw}&\quad \sum_{i=1}^M \sum_{j=1}^N p_i w_{ji}\log \frac{w_{ji}}{\sum_{i'=1}^M p_{i'} w_{ji'}}, \\
\text{s.t.}&\quad\sum_{j=1}^N w_{ji}=1,\quad\forall i; \\
&\quad\sum_{i=1}^M \sum_{j=1}^N \sum_{k=1}^K p_iw_{ji}s_{ki}\log \frac{\sum_{i'=1}^M p_{i'}w_{ji'}s_{ki'}}{\sum_{i'=1}^M p_{i'}w_{ji'}}\geq\hat{I}. 
\end{align}
\end{subequations}

The above problem is hard to solve due to the intricate logarithmic terms involving the Markov chain and transition probability conditions within the objective function and constraints. 
To address the aforementioned difficulty, a recent work \cite{Chen23IBOT} proposed to solve the RI function as a constraint optimization problem in a higher dimensional space, by introducing auxiliary variables $r_{j}=P_T(t_j)$ and $z_{kj}=P_{Y|T}(y_k|t_j)$ to transform the Markov chain and the transition probability conditions into additional constraints, thereby simplifying the objective function.
Then, the RI function can be written as
\begin{subequations}\label{original_RI}
\begin{align}
\min_{\bdw,\bdr,\bdz}&\quad \sum_{i=1}^M \sum_{j=1}^N p_i w_{ji}(\log w_{ji}-\log r_j), \\
\text{s.t.}&\quad\sum_{j=1}^N w_{ji}=1,\quad\forall i;\quad\sum_{k=1}^K z_{kj}=1,\quad\forall j; \\
&\quad\sum_{i=1}^M p_i w_{ji}=r_j,\quad\forall j;\quad\sum_{j=1}^N r_j=1; \\
&\quad\sum_{i=1}^M p_i w_{ji}s_{ki}\Big/\Big(\sum_{i=1}^M p_i w_{ji}\Big)=z_{kj},\quad\forall j,k; \\
&\quad\sum_{i=1}^M \sum_{j=1}^N \sum_{k=1}^K p_iw_{ji}s_{ki}\log z_{kj}\geq\hat{I}. 
\end{align}
\end{subequations}

\section{Semi-Relaxed IB Model}

Traditional methods, like the ADMM-based ALM method \cite{huang_2021_admm,bayat2019_ib_revisited} and the OT-based method \cite{Chen23IBOT}, solve the RI function by considering the Markov chain and transition probability condition as constraints of variables.
However, these methods exhibit limitations in computational efficiency when designing algorithms, due to the consideration of these constraints when analyzing the Lagrangian. 

To simplify the formulation, we construct a novel form for the IB problem in \eqref{original_RI} by relaxing some constraints that are implicitly contained in the model.
In light of this development, we delve into the original motivation behind considering the problem's formulation in a higher dimension space to obtain a simple structure and relaxing the Markov chain and transition probability constraints, leading us to propose a semi-relaxed IB model, as follows:
\begin{subequations} \label{SR_RI}
\begin{align}
\min_{\bdw,\bdr,\bdz}&\quad f(\bdw,\bdr)\overset{\triangle}{=}\sum_{i=1}^M \sum_{j=1}^N p_i w_{ji}(\log w_{ji}-\log r_j), \label{SR_RI_1} \\
\text{s.t.}&\quad\sum_{j=1}^N w_{ji}=1,\quad\forall i;\quad\sum_{k=1}^K z_{kj}=1,\quad\forall j; \label{SR_RI_2} \\
&\quad\sum_{j=1}^N r_j=1; \quad\sum_{i=1}^M \sum_{j=1}^N \sum_{k=1}^K p_iw_{ji}s_{ki}\log z_{kj}\geq\hat{I}. \label{SR_RI_3}
\end{align}
\end{subequations}

This formulation not only simplifies the structure of mutual information constraints and the objective function, but also reduces redundant constraints in the problem formulation when relaxing the Markov chain and the transition probability constraints from the IB model \eqref{original_RI}.
Besides, we can prove the equivalence of the solution for both the semi-relaxed IB model \eqref{SR_RI} and the original IB model \eqref{original_RI}.
\begin{theorem}
The optimal solution to the semi-relaxed IB model \eqref{SR_RI} is exactly that to the original IB model \eqref{original_RI}.
\end{theorem}
\begin{proof}
Assuming $(\bdw^*,\bdr^*,\bdz^*)$ is the optimal solution in \eqref{SR_RI}, it thus satisfies the Karush-Kuhn-Tucker (KKT) condition.
For short, we denote the constraints in \eqref{SR_RI_2} and \eqref{SR_RI_3} as $g_1(\bdw)={\bm 0}$, $g_2(\bdz)={\bm 0}$, $g_3(\bdr)=0$ and $g_4(\bdw,\bdz)
\leq 0$, and then we have 
\begin{subequations} \label{KKT}
\begin{align}
&\nabla_{\bdr} f(\bdw^*, \bdr^*)+\eta\nabla g_1(\bdr^*) = {\bm 0}, \\
& \bdbeta\nabla_{\bdz} g_2(\bdz^*)+\lambda\nabla g_4(\bdw^*,\bdz^*) = {\bm 0},
\end{align}
\end{subequations}
where $\eta,\bdbeta$ and $\lambda\in\mbbR^+$ are the corresponding multipliers. 
Noting that $w^*_{ij}$ and $s_{ik}$ are the conditional probabilities, we can derive that $\eta=1$ and $\beta_j=\lambda\sum_{i=1}^M p_i w_{ji}^{*}$. 
Substituting these into condition \eqref{KKT}, we obtain $(\bdw^*,\bdr^*,\bdz^*)$ satisfying the relaxed conditions.
Hence, $(\bdw^*,\bdr^*,\bdz^*)$ is also a feasible solution to the original IB model \eqref{original_RI}.
%
%
In this way, we have shown that the optimal solution set to the semi-relaxed IB model \eqref{SR_RI} is a subset of that to the original IB model \eqref{original_RI}.
%
%
On the other hand, the optimal solution to \eqref{original_RI} is also a feasible solution to \eqref{SR_RI}. 
Since the objective functions are the same for these two models, it is also optimal to \eqref{SR_RI}. 
%
\end{proof}

\section{The Alternating Bregman Projection Algorithm}

In this section, we propose a convergence guaranteed algorithm for computing the proposed semi-relaxed IB model. 
Due to the introduction of the Bregman projection \cite{benamou2015iterative} at one main step and the alternating minimization framework, we name it as the Alternating Bregman Projection (ABP) algorithm. 
%

\subsection{Algorithm Derivation and Implementation}

By introducing the multipliers $\bdalpha\in\mathbb{R}^M$, $\bdbeta\in\mathbb{R}^N$, $\eta\in\mathbb{R}$, $\zeta\in\mathbb{R}^+$, the Lagrangian of the semi-relaxed IB model is
\begin{equation*}
\begin{aligned}
&\mathcal{L}(\bdw,\bdr,\bdz;\eta,\bdalpha,\bdbeta,\lambda)=\!\sum_{i=1}^M \sum_{j=1}^N p_i w_{ji}(\log w_{ji}-\log r_j) \\
&+\sum_{i=1}^M \alpha_i\Big(\sum_{j=1}^N w_{ji}-1\Big)+\sum_{j=1}^N \beta_j\Big(\sum_{k=1}^K z_{kj}-1\Big) \\
&+\eta\Big(\sum_{j=1}^N r_j\!-\!1\Big)\!\!+\!\!\lambda\Big(\!-\!\!\sum_{i=1}^M \sum_{j=1}^N \sum_{k=1}^K p_i w_{ji}s_{ki}\log z_{kj}+\hat{I}\Big).
\label{formula-lagrangian-ri}
\end{aligned}
\end{equation*}
%

Our key ingredient is optimizing the primal variables $\bdw,\bdr,\bdz$ in an alternative manner. 
%
%
Based on the Lagrangian of the semi-relaxed model, we take derivatives with respect to the primal variables $\bdw,\bdr$ and $\bdz$ for its optimal expression, yielding closed-form iterative solutions. 
As a result, the efficiency of computations is ensured during each iteration, and the descent value can be accurately estimated.

\subsubsection{Updating $w$ via its Dual Variables}

Taking partial derivative of $\mathcal{L}(\bdw,\bdr,\bdz;\eta,\bdalpha,\bdbeta,\lambda)$ with respect to $w$, and denoting $d_{ij}=-\sum_{k=1}^K s_{ki}\log z_{kj}$ as the metric for short, we obtain the first order condition
\begin{equation*}
\dfrac{\partial\mathcal{L}}{\partial w_{ji}}=p_i(1+\log w_{ji}-\log r_j)+\alpha_i+\lambda p_i d_{ij}=0,
\end{equation*}
which implies $w_{ji}=e^{-\frac{\alpha_i}{p_i}-1}e^{-\lambda d_{ij}}r_j$. 
Using the Bregman projection into a polytope, we obtain 
\begin{equation*}
w_{ji}=e^{-\lambda d_{ij}}r_j\Big/ \Big(\sum_{j'=1}^N e^{-\lambda d_{ij'}}r_{j'}\Big),\quad\forall i,j,
\end{equation*}
where $\lambda$ is the root of a monotonic function $G(\lambda)$ defined by
\begin{equation*}
G(\lambda):=\!\!\sum_{i=1}^M\sum_{j=1}^N d_{ij}p_i\Bigg(e^{-\lambda d_{ij}}r_j \Big/ \Big(\sum_{j'=1}^N e^{-\lambda d_{ij'}}r_{j'}\Big)\Bigg)+\hat{I}.
\end{equation*}

\subsubsection{Updating $\bdr$ via its Dual Variable}
%
%

%
The first order condition yields 
\begin{equation*}
\dfrac{\partial\mathcal{L}}{\partial r_j}=-\sum_{i=1}^M\dfrac{p_i w_{ji}}{r_j}+\eta=0, 
\end{equation*}
which implies $r_j=\Big(\sum_{i=1}^M p_i w_{ji}\Big)\Big/\eta$.
Substituting the obtained expression into the constraint of $\bdr$, we have
\begin{equation*}
\sum_{j=1}^N\Big(\sum_{i=1}^M p_i w_{ji}\Big)\Big/\eta=1. 
\end{equation*}
Note that $\sum_{i=1}^M \sum_{j=1}^N p_i w_{ji}=1$, it implies $\eta=1$. 
Then we can update $\bdr$ under fixed $\bdw$ and $\bdz$ by $r_j=\sum_{i=1}^M p_i w_{ji},\forall j$. 
This is actually the relaxed transition probability constraint.

\subsubsection{Updating $\bdz$ via its Dual Variables}
%
%

The first order condition yields 
\begin{equation*}
\dfrac{\partial\mathcal{L}}{\partial z_{kj}}=\beta_j-\lambda\sum_{i=1}^M\dfrac{p_i w_{ji}s_{ki}}{z_{kj}}=0,
\end{equation*}
which implies $z_{kj}=\lambda\Big(\sum_{i=1}^M p_i w_{ji}s_{ki}\Big)\Big/\beta_j$.
Substituting this expression into the constraint of $\bdz$, we have
\begin{equation*}
\sum_{k=1}^K\lambda\Big(\sum_{i=1}^M p_i w_{ji}s_{ki}\Big)\Big/\beta_j=1.
\end{equation*}
Considering that $\sum_{k=1}^K s_{ki}=1$ implies $\beta_j=\lambda\sum_{i=1}^M p_i w_{ji}$, we can update $\bdz$ under fixed $\bdw$ and $\bdr$ by 
\begin{equation*}
z_{kj}=\Big(\sum_{i=1}^M p_i w_{ji}s_{ki}\Big)\Big/\Big(\sum_{i=1}^M p_i w_{ji}\Big),\quad \forall j,k.
\end{equation*}
This is actually the relaxed Markov chain condition.

After one iteration cycle, we finally update the metric $\bdd$ by its definition, i.e., $d_{ij}=-\sum_{k=1}^K s_{ki}\log z_{kj}$ after obtaining $\bdz$. 

For clarity, we summarize the procedure in Algorithm~\ref{alg:ri}.

\begin{algorithm}
\caption{The proposed ABP algorithm}
\label{alg:ri}
\begin{algorithmic}[ht]
\STATE{\bf Input} $p_i=p(x_i)$, $s_{ki}=p(y_k|x_i)$, $R$, \textit{max\_iter}
\STATE{\bf Output} Minimal $\sum_{i=1}^M\sum_{i=1}^N p_iw_{ji}(\log w_{ji}-\log r_j)$
\STATE Randomly initialize $\bdw^{(0)}$, $\bdr^{(0)}$, $\bdz^{(0)}$ and $\bdd^{(0)}$, $\lambda^{(0)}\leftarrow 1$
\FOR{$n=1:\textit{max\_iter}$}
    \STATE Find the root $G(\lambda^{(n)})=0$ by Newton's method
    \STATE $w_{ji}^{(n)}\leftarrow\!\!\Big(e^{-\lambda d_{ij}^{(n-1)}}r_j^{(n-1)}\Big)\!\!\Big/\!\!\Big(\sum_{j'=1}^N e^{-\lambda d_{ij'}^{(n-1)}}\!\!r_{j'}^{(n-1)}\Big)$
    \STATE $r_j^{(n)}\leftarrow\sum_{i=1}^M p_i w_{ji}^{(n)}$
    \STATE $z_{kj}^{(n)}\leftarrow\Big(\sum_{i=1}^M p_i w_{ji}^{(n)}s_{ki}\Big)\Big/\Big(\sum_{i=1}^M p_i w_{ji}^{(n)}\Big)$
    \STATE $d_{ij}^{(n)}\leftarrow-\sum_{k=1}^K s_{ki}\log z_{kj}^{(n)}$
\ENDFOR
\STATE{\bf Return} $\sum_{i=1}^M\sum_{i=1}^N p_i w_{ji}^{(n)}(\log w_{ji}^{(n)}-\log r_j^{(n)})$
\end{algorithmic}
\end{algorithm}

Although the ABP algorithm shares some similarities with the well-established BA algorithm, their conceptual frameworks differ significantly. 
%
%
Specifically, the semi-relaxed IB model is considered as a constraint optimization in the ABP algorithm, rather than treating it solely as an unconstrained objective function in the BA algorithm.  
This perspective allows for a more flexible and efficient update scheme for the Lagrange multiplier, which in turn leads to improved convergence properties and computational efficiency. 
In contrast, the BA algorithm cannot compute $R(I)$ directly with a given $I$.
%


Moreover, a pivotal distinction between the ABP algorithm for the semi-relaxed IB model and the GAS algorithm \cite{Chen23IBOT} for the IB-OT model lies in the relaxation of the Markov chain and transition probability constraints. 
Nevertheless, the proposed algorithm is able to recover the relaxed Markov chain and transition probability constraints, which guarantees the feasibility of the iteration points when solving the semi-relaxed model in a higher dimensional space. 
%
%
%
%
%
This relaxation gives rise to an innovative iterative scheme, wherein each variable can be updated through a closed-form solution, ensuring great efficiency and avoiding solving sub-problems. 
In addition, the dual problem with respect to $\bdw$ can be solved accurately, leading to a provably convergent algorithm, since the descent can be estimated in each direction. 
%

\addtolength{\topmargin}{0.04in}

\subsection{Convergence Analysis}

In this subsection, we present the convergence analysis of the proposed ABP algorithm for the semi-relaxed IB model in computing the RI function. 
It is worth noting that the descent of objective function between adjacent iterations can be obtained, due to the closed-form expressions in our ABP algorithm, which is the foundation of the convergence analysis.
%
%

\begin{lemma}
The objective function $f(\bdw,\bdr)$ is non-increasing during each iteration step, i.e., 
\begin{equation*}
f(\bdw^{(n)},\bdr^{(n)})\!\leq\! f(\bdw^{(n)},\bdr^{(n-1)})\!\leq\! f(\bdw^{(n-1)},\bdr^{(n-1)}).
\end{equation*}
Moreover, we have a descent estimation of the objective between adjacent iterations: 
\begin{multline}\label{err_estimate}
f(\bdw^{(n)},\bdr^{(n)})-f(\bdw^{(n-1)},\bdr^{(n-1)}) \\
=-\lambda^{(n)}\sum_{j=1}^N r_j^{(n-1)}D(\bdz_j^{(n-1)}\Vert\bdz_j^{(n-2)}) \\
-D(\bdr^{(n)}\Vert\bdr^{(n-1)})-\sum_{i=1}^M p_i D(\bdw_i^{(n-1)}\Vert\bdw_i^{(n)}).
\end{multline}
Here vector $\bdw_i$ denotes the $i$-th column of the matrix $\bdw$, and vector $\bdz_j$ denotes the $j$-th column of the matrix $\bdz$. 
\end{lemma}
\begin{proof}
We calculate the descent in two steps. First, we have
\begin{multline*}
f(\bdw^{(n)},\bdr^{(n)})-f(\bdw^{(n)},\bdr^{(n-1)}) \\
=\sum_{i=1}^M\sum_{j=1}^N p_i w_{ji}^{(n)}\log\frac{r_j^{(n-1)}}{r_j^{(n)}} =-D(\bdr^{(n)}\Vert \bdr^{(n-1)})\leq 0,
\end{multline*}
which corresponds to the update scheme of $\bdr$. 
Second, according to the update rule of $\bdw$, we obtain
\begin{equation*}
-\hat{I}=\sum_{i=1}^M \sum_{j=1}^N p_i w_{ji}^{(n)}d_{ij}^{(n-1)}=\sum_{i=1}^M \sum_{j=1}^N p_i w_{ji}^{(n-1)}d_{ij}^{(n-2)}.
\end{equation*}
Next, we have
\begin{multline*}
f(\bdw^{(n)},\bdr^{(n-1)})-f(\bdw^{(n-1)},\bdr^{(n-1)}) \\
=\sum_{i=1}^M \sum_{j=1}^N p_i w_{ji}^{(n)}\log\dfrac{w_{ji}^{(n)}}{r_j^{(n-1)}}+\lambda^{(n)}\sum_{i=1}^M \sum_{j=1}^N p_i w_{ji}^{(n)}d_{ij}^{(n-1)} \\
-\sum_{i=1}^M \sum_{j=1}^N p_i w_{ji}^{(n-1)}\log\dfrac{w_{ji}^{(n-1)}}{r_j^{(n-1)}} -\lambda^{(n)}\sum_{i=1}^M \sum_{j=1}^N p_i w_{ji}^{(n-1)}d_{ij}^{(n-2)} \\
=\sum_{i=1}^M \sum_{j=1}^N p_i w_{ji}^{(n-1)}\log\Big(1\Big/\big(\sum_{j=1}^N e^{-\lambda^{(n)}d_{ij}^{(n-1)}}r_j^{(n-1)}\big)\Big) \\
-\sum_{i=1}^M \sum_{j=1}^N p_i w_{ji}^{(n-1)}\log\Big(w_{ji}^{(n-1)}\Big/\big(e^{-\lambda^{(n)}d_{ij}^{(n-1)}}r_j^{(n-1)}\big)\Big) \\
-\lambda^{(n)}\sum_{i=1}^M \sum_{j=1}^N \sum_{k=1}^K p_i w_{ji}^{(n-1)}s_{ki}\log(z_{kj}^{(n-1)}/z_{kj}^{(n-2)}).
\end{multline*}
The last equation is from $\sum_{j=1}^N w_{ji}^{(n)}=\sum_{j=1}^N w_{ji}^{(n-1)}=1$ and the substitution of $d_{ij}^{(n)}=-\sum_{k=1}^Ks_{ki}\log z_{kj}^{(n)}$. 
Finally, one has
\begin{multline*}
f(\bdw^{(n)},\bdr^{(n-1)})-f(\bdw^{(n-1)},\bdr^{(n-1)}) \\
=\sum_{i=1}^M \sum_{j=1}^N p_i w_{ji}^{(n-1)}\log(w_{ji}^{(n)}/w_{ji}^{(n-1)}) \\
-\lambda^{(n)}\sum_{i=1}^M \sum_{j=1}^N \sum_{k=1}^K p_i w_{ji}^{(n-1)} z_{kj}^{(n-1)}\log(z_{kj}^{(n-1)}/z_{kj}^{(n-2)}) \\
=-\sum_{i=1}^M p_i D(\bdw_i^{(n-1)}\Vert \bdw_i^{(n)}) \\
-\lambda^{(n)}\sum_{j=1}^N r_j^{(n-1)}D(\bdz_j^{(n-1)}\Vert \bdz_j^{(n-2)}) \leq 0.
\end{multline*}
This step corresponds to the update schemes of $\bdw$ and $\bdz$. 
Therefore, the objective is non-increasing in each iteration, and the estimation (\ref{err_estimate}) is obtained.
\end{proof}

\begin{lemma}
The objective $f(\bdw,\bdr)$ is non-negative.  
\end{lemma}
\begin{proof}
$\sum_{j=1}^N r_j=1$ yields $\log r_j\leq 0$, and thus
\begin{multline*}
f(\bdw,\bdr)=\sum_{i=1}^M \sum_{j=1}^N p_i w_{ji}\log\dfrac{w_{ji}}{r_j}\geq\sum_{i=1}^M \sum_{j=1}^N p_i w_{ji}\log w_{ji} \\
\geq\sum_{i=1}^M p_i\big(\sum_{j=1}^N w_{ji}\big)\log\big(\sum_{j=1}^N w_{ji}\big)=0.
\end{multline*}
Hence, the objective $f(\bdw,\bdr)$ is lower bounded by zero.
\end{proof}

Since the objective function $f(\bdw,\bdr)$ is non-increasing and lower bounded, it must converge to a local minimum. 
Moreover, the descent estimation ensures the convergence of the iterative sequence. 

\begin{theorem}
The sequence $\big\{(\bdw^{(n)},\bdr^{(n)},\bdz^{(n)})\big\}$ converges to some local minimum $(\bdw^{*},\bdr^{*},\bdz^{*})$. 
\end{theorem}
\begin{proof}
From Pinsker's inequality and \eqref{err_estimate}, we have
\begin{multline*}
f(\bdw^{(n-1)},\bdr^{(n-1)})-f(\bdw^{(n)},\bdr^{(n)}) \\
\geq \frac{1}{2}\left[\lambda^{(n)}\sum_{j=1}^N r_j^{(n-1)}\Vert\bdz_j^{(n-1)}-\bdz_j^{(n-2)}\Vert_1^2 \right.\\
\left.+\Vert\bdr^{(n)}-\bdr^{(n-1)}\Vert_1^2+\sum_{i=1}^M p_i\Vert\bdw_i^{(n-1)}-\bdw_i^{(n)}\Vert_1^2\right] \geq 0, 
\end{multline*}
where $\Vert\cdot\Vert_{1}$ denotes the $L_1$-norm. 
Since $\{f(\bdw^{(n)},\bdr^{(n)})\}$ converges, we have
\begin{equation*}
\sum_{n=1}^\infty\big(f(\bdw^{(n)},\bdr^{(n)})-f(\bdw^{(n+1)},\bdr^{(n+1)})\big)<\infty.
\end{equation*}
The above bound gives $\sum_{n=1}^\infty\Vert\bdr^{(n)}-\bdr^{(n-1)}\Vert<\infty$, and thus $\{\bdr^{(n)}\}$ converges. 
Similarly we have $\{\bdw^{(n)}\}$ converges and the iteration formula also yields that $\{\bdz^{(n)}\}$ converges. 
%
%
Hence, we denote $(\bdw^{*},\bdr^{*},\bdz^{*})$ as the limitation of the sequence $(\bdw^{(n)},\bdr^{(n)},\bdz^{(n)})$. 
By taking the limit of the iterative expressions in the ABP algorithm, we obtain that $(\bdw^{*},\bdr^{*},\bdz^{*})$ satisfies the KKT conditions. 
In fact, the iterations ensure the feasibility of this limit point, and the inequality constraint is satisfied by the multiplier $\lambda^{*}$, confirming that the limit point meets the KKT conditions.
\end{proof}

\section{Numerical Results}

This section evaluates the efficiency and effectiveness of the proposed approach by conducting numerical experiments over classical examples. 
These experiments have been implemented by Matlab R2022a on a Linux platform with 128G RAM and one Intel(R) Xeon(R) Gold 5117 CPU@2.00GHz. 

\subsection{Accuracy and Efficiency on Classical Distributions}

This subsection computes the relevance-compression function $R(I)$ of two classical models, i.e. the jointly Bernoulli model and the jointly Gaussian model formulated in \cite{zaidi_2020_models}.

For the jointly Bernoulli model, we set $X\sim B(0,1/2)$, $Y\sim B(0,1/2)$ and $X\oplus Y\sim B(0,e)$, where the summation is defined modulo $2$, $0\leq e\leq 1/2$. 
Let $h$ be the entropy, i.e., 
$
h(x):=-x\log x-(1-x)\log(1-x).
$
If $I=\log 2-h(u)$ where $e\leq u\leq 1/2$ is a parameter, then $R(I)=\log 2-h(v)$, $v=e+(1-2e)u$.
Here, we take the flip probability $e=0.15$, and the dimension is set as $M=N=K=2$.

For the jointly Gaussian model, we set $X=\sqrt{\text{SNR}}~Y+S$ where $X,S\sim N(0,1)$, then
\begin{equation*}
R(I)=-\dfrac{1}{2}\log\dfrac{(1+\text{SNR})e^{-2I}-1}{\text{SNR}}.
\end{equation*}
In our experiment, we truncate the variables into an interval $[-L,L]$, and discretize the interval by a set of uniform grid points $\{x_i\}_{i=1}^M$, where $x_i=-L+(i-1/2)\delta, \delta=2L/M$.
The parameter is set as $\text{SNR}\!=\!1$, $L\!=\!10$ and $M\!=\!N\!=\!K\!=\!100$.

We compare the performance of ABP algorithm with that of traditional BA algorithm. 
Moreover, we extend our comparison with the GAS algorithm \cite{Chen23IBOT} which directly computes the RI function as well as the ADMM based method for the ALM problem of the RI function \cite{huang_2021_admm}. 
The computational results are summarized in Table I, wherein each result has been obtained by repeating the experiment $50$ times. 
Due to numerical instability problems encountered while employing the ADMM algorithm \cite{huang_2021_admm} for computing the jointly Gaussian model, the running time is not recorded in Table I. 
%
%
\footnote{We utilized the code at \textit{“https://github.com/hui811116/ib-admm”} to implement the ADMM algorithm. 
%
%
Since the scale of jointly Bernoulli model is limited ($M=N=K=2$), we successfully obtain numerical results reported in Table I. 
However, the scale of our jointly Gaussian model is much too larger ($M=N=K=100$) for ADMM algorithm to execute normally, which may be due to numerical instability in gradient computations. 
Based on the results of all tested cases, we believe that the ABP algorithm could still be faster than the ADMM algorithm on a suitable scale.
}

\begin{table}[ht]
    \renewcommand\arraystretch{1.2}
    \centering
    \caption{Comparison among ABP algorithm, GAS algorithm, BA algorithm and ADMM algorithm for classical distributions.}
    \begin{tabular}{|c|c|c|c|c|c|} 
    \hline
        \multirow{2}{*}{} & \multirow{2}{*}{$(I,\lambda_I)$} & 
        \multicolumn{4}{c|}{Time($\times 10^{-2}$s)} \\ \cline{3-6}
        \multicolumn{1}{|c|}{} & \multicolumn{1}{c|}{} & ABP & GAS & BA & ADMM \\ \cline{1-6}
        \multirow{4}{*}{Bernoulli} & $(0.0626,2.1478)$ & $0.0781$ & $0.402$ & $1.83$ & $0.703$ \\  
        \multicolumn{1}{|c|}{} & $(0.0910,2.2118)$ & $0.0419$ & $0.114$ & $0.926$ & $0.291$ \\
        \multicolumn{1}{|c|}{} & $(0.1254,2.3113)$ & $0.0884$ & $0.307$ & $0.764$ & $0.208$ \\ 
        \multicolumn{1}{|c|}{} & $(0.1662,2.4800)$ & $0.0275$ & $0.028$ & $0.406$ & $0.396$ \\ \cline{1-6}
        \multirow{5}{*}{Gaussian} & $(0.0400,2.1817)$ & $0.630$ & $20.8$ & $99.3$ & $-$ \\ 
        \multicolumn{1}{|c|}{} & $(0.0800,2.4199)$ & $1.492$ & $18.2$ & $95.3$ & $-$ \\ 
        \multicolumn{1}{|c|}{} & $(0.1200,2.7444)$ & $1.245$ & $25.2$ & $109.6$ & $-$ \\ 
        \multicolumn{1}{|c|}{} & $(0.1600,3.2109)$ & $2.533$ & $25.8$ & $108.8$ & $-$ \\ 
        \multicolumn{1}{|c|}{} & $(0.2000,3.9357)$ & $4.828$ & $33.1$ & $57.8$ & $-$ \\ \cline{1-6}
    \end{tabular}
    \vspace{+.05in}
    
    \footnotesize{Notes: a) For Bernoulli case, we take $I$'s value corresponding to $u=0.1,0.15,0.2,0.25$; b) The algorithm stops if the decrease of $R$ with respect to the previous step is within $10^{-6}$; c) For the BA and ADMM algorithm, we search for the suitable $\lambda_I$ adaptively if the error is within $10^{-4}$; it generally takes around $20$ trials to find a suitable $\lambda_I$.}
\end{table}

As shown in Table I, the ABP algorithm has a significant advantage against traditional BA algorithm in computing time for both models, resulting in the speed-up ratios reaching tens and more. 
Additionally, our ABP algorithm also demonstrates superiority in computing time over the GAS algorithm and the ADMM algorithm with the same accuracy. 

\subsection{Convergence Behavior and Algorithm Verification}

This subsection verifies the convergence behavior of the ABP algorithm by evaluating the summation of residual errors in the $L_1$-norm within the KKT condition. 
Specifically, we focus on the case where ABP algorithm is applied to compute RI function of the jointly Bernoulli and Gaussian model with different value $I$. 
In these two experiments, parameters are set the same as the above, and the maximum number of iteration is set as $3000$. 
As shown in Fig. \ref{fig_convergence}, the residual errors converge below $10^{-9}$ in all tested cases. 
\begin{figure}[htbp]
    \centering
    \includegraphics[width=0.49\linewidth]{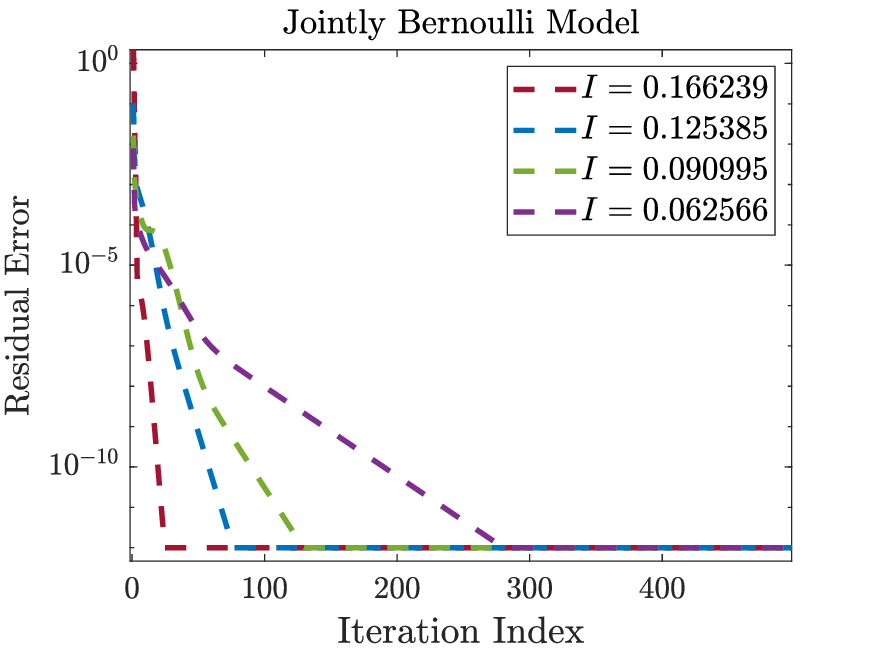}
    \includegraphics[width=0.49\linewidth]{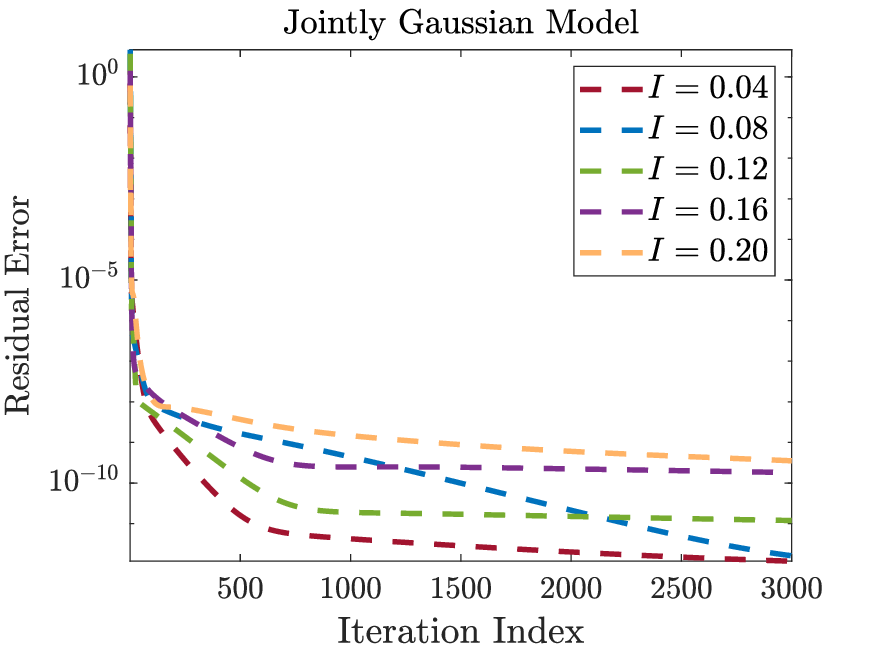} 
    \caption{The convergent trajectories of the residual error for the proposed ABP algorithm. Bernoulli (Left), Gaussian (Right).} 
    \label{fig_convergence} 
\end{figure}

\subsection{Experiments on Iris Dataset}
In this subsection, we conduct the ABP algorithm on a real-world classification dataset: the Iris dataset from the UCI learning repository\cite{blake_1998_uci}. 
We regard the observable variable $X$ as the sample, and the target variable $Y$ as the classification results. 
The joint distribution is set as the empirical distribution derived from the Iris dataset. 
Specifically, we discretize the product feature space in units of $1$, count the frequency of samples falling into each unit, and finally obtain the joint empirical distribution after filtering zeros.

\begin{figure}[ht]
    \centering
    \includegraphics[width=0.45\textwidth]{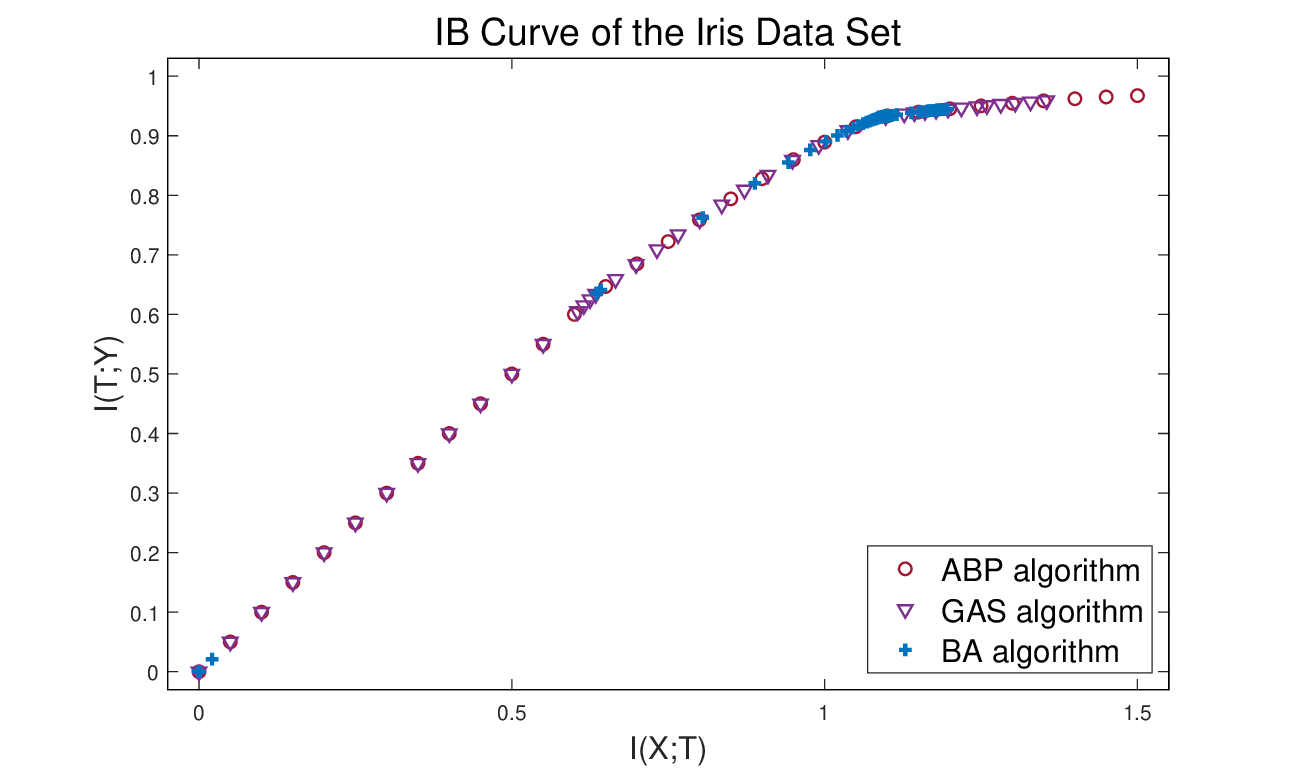}
    \caption{Comparison among the ABP, GAS and BA algorithm for a scenario of a real-world dataset in classification task.}
    \label{fig_iris}
\end{figure}

%
%
As illustrated in Fig. \ref{fig_iris}, both the ABP and GAS algorithms successfully reconstruct the IB curve for this task. 
Notably, the ABP algorithm demonstrates robustness even when the threshold $I$ is very high. 
%
%
In addition, the ABP algorithm exhibits excellent numerical stability performance around points with phase transition phenomena \cite{kolchinsky2018ib_caveats} and offers substantial advantages in terms of computation time over other algorithms.
In contrast, numerical results of many thresholds $I$ are missing by the BA algorithm. 
We also utilize the open source code of the ADMM to compute this task, but similar to before, some numerical errors may occur when the dimensions of $X$ and $Y$ are different, so the corresponding results are not plotted.

\section{Conclusion}
In this paper, we propose a semi-relaxed IB model, where the Markov chain and transitional probability condition are relaxed from the proposed model.
Based on this model, we design an efficient algorithm that turns out to recover the semi-relaxed constraints in the model with theoretical convergence guaranteed, by analyzing the Lagrangian of the relaxed model.
Numerical experiments demonstrate notable advantages of our proposed method compared with other existing approaches.

\bibliographystyle{bibtex/IEEEtran}
\bibliography{bibtex/bib/ref}

\begin{thebibliography}{10}
\providecommand{\url}[1]{#1}
\csname url@samestyle\endcsname
\providecommand{\newblock}{\relax}
\providecommand{\bibinfo}[2]{#2}
\providecommand{\BIBentrySTDinterwordspacing}{\spaceskip=0pt\relax}
\providecommand{\BIBentryALTinterwordstretchfactor}{4}
\providecommand{\BIBentryALTinterwordspacing}{\spaceskip=\fontdimen2\font plus
\BIBentryALTinterwordstretchfactor\fontdimen3\font minus \fontdimen4\font\relax}
\providecommand{\BIBforeignlanguage}[2]{{%
\expandafter\ifx\csname l@#1\endcsname\relax
\typeout{** WARNING: IEEEtran.bst: No hyphenation pattern has been}%
\typeout{** loaded for the language `#1'. Using the pattern for}%
\typeout{** the default language instead.}%
\else
\language=\csname l@#1\endcsname
\fi
#2}}
\providecommand{\BIBdecl}{\relax}
\BIBdecl

\bibitem{tishby_1999_ib}
N.~Tishby, F.~C. Pereira, and W.~Bialek, ``{The Information Bottleneck Method},'' \emph{Proc. 37th Annual Allerton Conference on Communications, Control and Computing}, 1999.

\bibitem{cover_2003_elements}
T.~Cover, J.~Thomas, and J.~Wiley, ``{Elements of Information Theory},'' \emph{Tsinghua University Pres}, 2003.

\bibitem{tishby_2003_tradeoff}
R.~Gilad-Bachrach, A.~Navot, and N.~Tishby, ``An information theoretic tradeoff between complexity and accuracy,'' in \emph{Learning Theory and Kernel Machines: 16th Annual Conference on Learning Theory and 7th Kernel Workshop, COLT/Kernel 2003, Washington, DC, USA, August 24-27, 2003. Proceedings}.\hskip 1em plus 0.5em minus 0.4em\relax Springer, 2003, pp. 595--609.

\bibitem{hassanpour_2020_forwardaware}
S.~Hassanpour, T.~Monsees, D.~W{\"u}bben, and A.~Dekorsy, ``{Forward-Aware Information Bottleneck-Based Vector Quantization for Noisy Channels},'' \emph{IEEE Transactions on Communications}, vol.~68, no.~12, pp. 7911--7926, 2020.

\bibitem{chen_2016_alternating}
D.~Chen and V.~Kuehn, ``{Alternating Information Bottleneck Optimization for Weighted Sum Rate and Resource Allocation in the Uplink of C-RAN},'' in \emph{WSA 2016; 20th International ITG Workshop on Smart Antennas}.\hskip 1em plus 0.5em minus 0.4em\relax VDE, 2016, pp. 1--7.

\bibitem{zaidi_2020_models}
A.~Zaidi, I.~E. Aguerri, and S.~Shamai, ``{On the Information Bottleneck Problems: Models, Connections, Applications and Information Theoretic Views},'' \emph{CoRR}, vol. abs/2002.00008, 2020.

\bibitem{shamir_2010_generalization}
O.~Shamir, S.~Sabato, and N.~Tishby, ``{Learning and Generalization with the Information Bottleneck},'' \emph{Theoretical Computer Science}, vol. 411, no. 29–30, pp. 2696--2711, 2010.

\bibitem{tishby_2017_dnn}
R.~Shwartz-Ziv and N.~Tishby, ``{Opening the Black Box of Deep Neural Networks via Information},'' \emph{arXiv preprint arXiv:1703.00810}, 2017.

\bibitem{alemi2016variation_ib}
A.~A. Alemi, I.~Fischer, J.~V. Dillon, and K.~Murphy, ``{Deep Variational Information Bottleneck},'' in \emph{Proc. 5th International Conference on Learning Representations (ICLR)}, Toulon, France, Apr. 2017, pp. 1--5.

\bibitem{kolchinsky2018ib_caveats}
A.~Kolchinsky, B.~D. Tracey, and S.~Van~Kuyk, ``{Caveats for Information Bottleneck in Deterministic Scenarios},'' in \emph{Proc. 7th International Conference on Learning Representations (ICLR)}, New Orleans, Louisiana, USA, May 2019, pp. 1--23.

\bibitem{wu2019phase}
T.~Wu and I.~Fischer, ``Phase transitions for the information bottleneck in representation learning,'' in \emph{International Conference on Learning Representations}, 2019.

\bibitem{huang_2021_admm}
T.-H. Huang and A.~El~Gamal, ``{A Provably Convergent Information Bottleneck Solution via ADMM},'' \emph{2021 IEEE International Symposium on Information Theory (ISIT)}, pp. 43--48, 2021.

\bibitem{boyd2011distributed}
S.~Boyd, N.~Parikh, E.~Chu, B.~Peleato, J.~Eckstein \emph{et~al.}, ``Distributed optimization and statistical learning via the alternating direction method of multipliers,'' \emph{Foundations and Trends{\textregistered} in Machine learning}, vol.~3, no.~1, pp. 1--122, 2011.

\bibitem{Chen23IBOT}
L.~Chen, S.~Wu, W.~Ye, H.~Wu, H.~Wu, W.~Zhang, B.~Bai, and Y.~Sun, ``{Information Bottleneck Revisited: Posterior Probability Perspective with Optimal Transport},'' in \emph{2023 IEEE International Symposium on Information Theory (ISIT)}, Taipei, Taiwan, China, Jun. 2023, pp. 1490--1495.

\bibitem{wu_2023_communication}
S.~Wu, W.~Ye, H.~Wu, H.~Wu, W.~Zhang, and B.~Bai, ``{A Communication Optimal Transport Approach to the Computation of the Rate Distortion Functions},'' \emph{Proc. 2023 IEEE Information Theory Workshop (ITW)}, 2023.

\bibitem{curuti_2013_sinkhorn}
M.~Cuturi, ``{Sinkhorn Distances: Lightspeed Computation of Optimal Transport},'' \emph{Advances in Neural Information Processing Systems}, vol.~26, 2013.

\bibitem{amjad2019classification}
R.~A. Amjad and B.~C. Geiger, ``{Learning Representations for Neural Network-based Classification Using the Information Bottleneck Principle},'' \emph{IEEE Transactions on Pattern Analysis and Machine Intelligence}, vol.~42, no.~9, pp. 2225--2239, Apr. 2019.

\bibitem{goldfeld_2020_ml}
Z.~Goldfeld and Y.~Polyanskiy, ``{The Information Bottleneck Problem and Its Applications in Machine Learning},'' \emph{IEEE Journal on Selected Areas in Information Theory}, vol.~1, no.~1, pp. 19--38, 2020.

\bibitem{bayat2019_ib_revisited}
F.~Bayat and S.~Wei, ``Information bottleneck problem revisited,'' in \emph{2019 57th Annual Allerton Conference on Communication, Control, and Computing (Allerton)}.\hskip 1em plus 0.5em minus 0.4em\relax IEEE, 2019, pp. 40--47.

\bibitem{benamou2015iterative}
J.-D. Benamou, G.~Carlier, M.~Cuturi, L.~Nenna, and G.~Peyr{\'e}, ``Iterative bregman projections for regularized transportation problems,'' \emph{SIAM Journal on Scientific Computing}, vol.~37, no.~2, pp. A1111--A1138, 2015.

\bibitem{blake_1998_uci}
C.~L. Blake, ``{UCI Repository of Machine Learning Databases},'' \emph{http://www. ics. uci. edu/\~{} mlearn/MLRepository. html}, 1998.

\end{thebibliography}
\newpage 
\begin{appendix}
\subsection{Formulation and Algorithm for IR Case}

Our semi-relaxed variant model for the IR case is given by
\begin{subequations}
\begin{align}
\min\limits_{w,r,z}&\quad f_{IR}(\bdw,\bdz)\overset{\triangle}{=}-\sum\limits_{i=1}^M \sum\limits_{j=1}^N \sum\limits_{k=1}^K p_i w_{ij}s_{ki}\log z_{kj}, \\
\text{s.t.}&\quad\sum\limits_{j=1}^N w_{ij}=1,\quad\forall i;\quad\sum\limits_{k=1}^K z_{kj}=1,\quad\forall j; \\ 
&\quad\sum\limits_{j=1}^N r_j=1; \quad\sum\limits_{i=1}^M \sum\limits_{j=1}^N p_i w_{ij}(\log \frac{w_{ij}}{r_j})\leq R.
\end{align}\label{formula-relaxed-ir}
\end{subequations}

By taking the first order condition into the constraints, we can similarly define the monotonic function
\begin{equation*}
\begin{aligned}
G_{IR}(\lambda)=&-\sum\limits_{i=1}^M p_i\log\Big(\sum\limits_{j=1}^N e^{-\lambda d_{ij}}r_j\Big)
-\lambda\sum\limits_{i=1}^M\sum\limits_{j=1}^N p_i d_{ij}\cdot \\
&\Bigg(e^{-\lambda d_{ij}}r_j\Big/\Big(\sum\limits_{j'=1}^N e^{-\lambda d_{ij'}}r_{j'}\Big)\Bigg)-R.
\end{aligned}
\end{equation*}

We can solve the IR function by alternatively minimizing the primal variables using the corresponding Lagrangian.
Then, we obtain the ABP algorithm for computing the IR function, where the only change is the way of updating $\lambda$.
The proposed algorithm is demonstrated as Algorithm~\ref{alg:ir}. 
\begin{algorithm}
\caption{ABP algorithm for the IR function}
\label{alg:ir}
\begin{algorithmic}[ht]
\STATE{\bf Input} $p_i=p(x_i)$, $s_{ki}=p(y_k|x_i)$, $I$, \textit{max\_iter}
\STATE{\bf Output} Minimal $-\sum\limits_{i=1}^M \sum\limits_{j=1}^N \sum\limits_{k=1}^K p_i w_{ij}s_{ki}\log z_{kj}$
\STATE Randomly initialize $\bdw^{(0)}$, $\bdr^{(0)}$, $\bdz^{(0)}$ and $\bdd^{(0)}$, $\lambda^{(0)}\leftarrow 1$
\FOR{$n=1:\textit{max\_iter}$}
    \STATE Find the root $G_{IR}(\lambda^{(n)})=0$ by Newton's method 
    \STATE $w_{ij}^{(n)}\leftarrow\Big(e^{-\lambda d_{ij}^{(n-1)}}r_j^{(n-1)}\Big)\Big/\Big(\sum\limits_{j'=1}^N e^{-\lambda d_{ij'}^{(n-1)}}r_{j'}^{(n-1)}\Big)$
    \STATE $r_j^{(n)}\leftarrow\sum\limits_{i=1}^M p_i w_{ij}^{(n)}$
    \STATE $z_{kj}^{(n)}\leftarrow\Big(\sum\limits_{i=1}^M p_i w_{ij}^{(n)}s_{ki}\Big)\Big/\Big(\sum\limits_{i=1}^M p_i w_{ij}^{(n)}\Big)$
    \STATE $d_{ij}^{(n)}\leftarrow-\sum\limits_{k=1}^K s_{ki}\log z_{kj}^{(n)}$
\ENDFOR
\STATE{\bf Return} $-\sum\limits_{i=1}^M \sum\limits_{j=1}^N \sum\limits_{k=1}^K p_i w_{ij}^{(n)}s_{ki}\log z_{kj}^{(n)}$
\end{algorithmic}
\end{algorithm}

\subsection{Convergence Analysis}
In this section, we provide a convergence proof of our ABP algorithm on computing the $I(R)$ function, which is analogous to the $R(I)$ case. 
For convenience, we also denote $f_{IR}(\bdw,\bdz)=-\sum\limits_{i=1}^M\sum\limits_{i=1}^N\sum\limits_{i=1}^K p_i w_{ij}s_{ki}\log z_{kj}$ for analysis.

\begin{theorem}
For IR case, the objective $f_{IR}$ is non-increasing during each iteration step, that is 
\begin{equation*}
f_{IR}(\bdw^{(n)},\bdz^{(n)})\!\leq\! f_{IR}(\bdw^{(n)},\bdz^{(n-1)})\!\leq\! f_{IR}(\bdw^{(n-1)},\bdz^{(n-1)}).
\end{equation*}
\end{theorem}
\begin{proof}
For the first inequality, we have
\begin{equation*}
\begin{aligned}
&f_{IR}(\bdw^{(n)},\bdz^{(n)})-f_{IR}(\bdw^{(n)},\bdz^{(n-1)}) \\
&=-\sum\limits_{i=1}^M\sum\limits_{i=1}^N\sum\limits_{i=1}^K p_i w_{ij}^{(n)}s_{ki}\log(z_{kj}^{(n)}/z_{kj}^{(n-1)}) \\
&=-\sum\limits_{j=1}^N\sum\limits_{k=1}^K\Big(\sum\limits_{i=1}^M p_i w_{ij}^{(n)}\Big)z_{kj}^{(n)}\log(z_{kj}^{(n)}/z_{kj}^{(n-1)}) \\
&=-\sum\limits_{j=1}^N r_j^{(n)}D(\bdz_j^{(n)}\Vert \bdz_j^{(n-1)})\leq 0.
\end{aligned}
\end{equation*}
For the second inequality, according to the update rule of $\bdw$,
\begin{equation*}
R\!=\!\!\sum\limits_{i=1}^M\sum\limits_{j=1}^N p_i w_{ij}^{(n)}\log\dfrac{w_{ij}^{(n)}}{r_j^{(n-1)}}
\!=\!\!\sum\limits_{i=1}^M\sum\limits_{j=1}^N p_i w_{ij}^{(n-1)}\log\dfrac{w_{ij}^{(n-1)}}{r_j^{(n-2)}}.
\end{equation*}
So, we have
\begin{equation*}
\begin{aligned}
&f_{IR}(\bdw^{(n)},\bdz^{(n-1)})-f_{IR}(\bdw^{(n-1)},\bdz^{(n-1)}) \\
&=\sum\limits_{i=1}^M \sum\limits_{j=1}^N\left[ p_i w_{ij}^{(n)}d_{ij}^{(n-1)}+\dfrac{1}{\lambda^{(n)}} p_i w_{ij}^{(n)}\log\dfrac{w_{ij}^{(n)}}{r_j^{(n-1)}}\right] \\
&-\sum\limits_{i=1}^M\sum\limits_{j=1}^N \left[p_i w_{ij}^{(n-1)}d_{ij}^{(n-1)}-\dfrac{1}{\lambda^{(n)}} p_i w_{ij}^{(n-1)}\log\dfrac{w_{ij}^{(n-1)}}{r_j^{(n-2)}}\right] \\
&=\dfrac{1}{\lambda^{(n)}}\Big(\sum\limits_{i=1}^M\sum\limits_{j=1}^N p_i w_{ij}^{(n)}\log\dfrac{1}{\sum\limits_{j=1}^N e^{-\lambda^{(n)}d_{ij}^{(n-1)}}r_j^{(n-1)}} \\
&-\sum\limits_{i=1}^M\sum\limits_{j=1}^N p_i w_{ij}^{(n-1)}\log\dfrac{w_{ij}^{(n-1)}}{e^{-\lambda^{(n)}d_{ij}^{(n-1)}}r_j^{(n-2)}}\Big) \\
=&\dfrac{1}{\lambda^{(n)}}\!\!\!\left[\sum\limits_{i=1}^M\sum\limits_{j=1}^N p_i w_{ij}^{(n-1)}\log\dfrac{w_{ij}^{(n)}}{w_{ij}^{(n-1)}}
\!\!-\!\!\!\sum\limits_{j=1}^N r_j^{(n-1)}\log\dfrac{r_j^{(n-1)}}{r_j^{(n-2)}}\right].
\end{aligned}
\end{equation*}
The last equation is because $\sum\limits_{j=1}^N w_{ij}^{(n)}=\sum\limits_{j=1}^N w_{ij}^{(n-1)}=1$, and the substitution of 
\begin{equation*}
w_{ij}^{(n)}=e^{-\lambda^{(n)}d_{ij}^{(n-1)}}r_j^{(n-1)}\Big/\Big(\sum\limits_{j'=1}^N e^{-\lambda^{(n)}d_{ij'}^{(n-1)}}r_{j'}^{(n-1)}\Big).
\end{equation*}
Hence, we have
\begin{equation*}
\begin{aligned}
&0\geq f_{IR}(\bdw^{(n)},\bdz^{(n-1)})-f_{IR}(\bdw^{(n-1)},\bdz^{(n-1)}) \\
=&-\dfrac{1}{\lambda^{(n)}}\Big(\sum\limits_{i=1}^M p_i D(\bdw_i^{(n-1)}\Vert\bdw_i^{(n)})+D(\bdr^{(n-1)}\Vert\bdr^{(n-2)})\Big).
\end{aligned}
\end{equation*}
Hence, the objective is non-increase.
\end{proof}

\begin{theorem}
The objective is non-negative. 
\end{theorem}
\begin{proof}
Since $\sum\limits_{k=1}^K z_{kj}=1$, we have $\log z_{kj}\leq 0$, so $f_{IR}(\bdw,\bdz)=-\sum\limits_{i=1}^M\sum\limits_{i=1}^N\sum\limits_{i=1}^K p_i w_{ij}s_{ki}\log z_{kj}\geq 0$. 
\end{proof}

From the above two theorems, we have the objective is non-increasing and lower bounded, hence it converges to a local minimum. 
Moreover, a descent estimation of the objective is ensured: 
\begin{equation*}
\begin{aligned}
&f_{IR}(\bdw^{(n)},\bdz^{(n)})-f_{IR}(\bdw^{(n-1)},\bdz^{(n-1)}) \\
=&-\sum\limits_{j=1}^N r_j^{(n)}D(\bdz_j^{(n)}\Vert \bdz_j^{(n-1)})-\dfrac{1}{\lambda^{(n)}}\Big(D(\bdr^{(n-1)}\Vert\bdr^{(n-2)}) \\
&+\sum\limits_{i=1}^M p_i D(\bdw_i^{(n-1)}\Vert\bdw_i^{(n)})\Big).
\end{aligned}
\end{equation*}
Similarly, we show the local convergence property.

\end{appendix}

\end{document}